\documentclass[11pt]{article}
\usepackage{hyperref}
\usepackage{graphicx,psfrag,epsf}
\usepackage{enumerate}
\usepackage[round, sort]{natbib}
\usepackage{url} 
\usepackage{setspace}
\usepackage{amsmath, amsthm, amssymb}
\usepackage[mathscr]{euscript}
\usepackage{mathtools}
\usepackage{multirow}
\usepackage{relsize}
\usepackage{bm}
\usepackage{multicol}
\usepackage{diagbox}
\usepackage{placeins}
\usepackage{subfigure}
\usepackage{authblk}
\usepackage{caption} 
\usepackage{amsfonts}
\usepackage{environ}
\makeatletter
\let\@fnsymbol\@arabic
\makeatother
\usepackage[margin=1in]{geometry}
\usepackage{fullpage}
\usepackage{bbm}
\usepackage{times}
\usepackage[ruled,vlined]{algorithm2e}
\include{pythonlisting}
\theoremstyle{definition}
\newtheorem{definition}{Definition}
\newtheorem{theorem}{Theorem}
\newtheorem{proposition}{Proposition}

\newtheorem{remark}{Remark}
\newcolumntype{L}{>{\centering\arraybackslash}m{1.5cm}}

\begin{document}
\title{\bf Bayesian Neural Tree Models for Nonparametric Regression}
\author[a]{Tanujit Chakraborty\thanks{Corresponding author. Email: tanujit\textunderscore r@isical.ac.in}}
\author[b]{Gauri Kamat}
\author[a]{Ashis Kumar Chakraborty}
\affil[a]{Statistical Quality Control \& Operations Research Unit, Indian Statistical Institute, India.}
\affil[b]{Department of Biostatistics, Brown University, Providence, Rhode Island, USA.}
\date{}
\maketitle
\begin{abstract}
Frequentist and Bayesian methods differ in many aspects but share some basic optimal properties. In real-life prediction problems, situations exist in which a model based on one of the above paradigms is preferable depending on some subjective criteria. Nonparametric classification and regression techniques, such as decision trees and neural networks, have both frequentist (classification and regression trees (CART) and artificial neural networks) as well as Bayesian (Bayesian CART and Bayesian neural networks) approach to learning from data. In this paper, we present two hybrid models combining the Bayesian and frequentist versions of CART and neural networks, which we call the Bayesian neural tree (BNT) models. BNT model can simultaneously perform feature selection and prediction, are highly flexible, and generalize well in settings with limited training observations. We study the statistical consistency of the proposed approach and derive the optimal value of a vital model parameter. We provide some illustrative examples using a wide variety of standard regression data sets to show the superiority of the proposed models.
\end{abstract}

\section{Introduction}
Methodologies in nonparametric regression employ either a frequentist or a Bayesian approach to learning from data. The choice between the two paradigms is often philosophical and based on subjective judgments. Two models, namely decision trees and neural networks, have primarily been used in the frequentist setting, but have robust Bayesian counterparts. Classification and regression trees (CART) were introduced by \citet{breimanetal1984} for flexibly modeling the conditional distribution of an outcome variable given the predictors.
For a data set, a tree is grown by sequentially splitting its internal nodes, and then pruning the grown tree back to avoid overfitting \citep{loh2011}. The splitting rule for each node is based on the minimization of the mean squared error (MSE) in regression and Gini index in classification. The Bayesian approach to finding a `good' tree model entails specification of a prior distribution and stochastic search \citep{chipman1998,chipman2002bayesian}. The fundamental idea behind Bayesian CART (BCART) is to have the prior induce a posterior distribution that can guide a (posterior) stochastic search towards a promising tree model \citep{chipman2002bayesian}.
\par 
On the other hand, an artificial neural network (ANN) is an interconnected gathering of artificial neurons organized in layers \citep{horniketal1989}. A standard ANN model has three layers of nodes, namely input, hidden, and output layers, where nodes are neurons that use a nonlinear activation function (except for the input nodes). A backpropagation gradient descent algorithm is used to compare the network outputs with the actual outputs \citep{rumelhardtetal1988}. If an error exists, it is backpropagated through the network, and the weights in the network architecture are adjusted accordingly \citep{lecun2015deep}. An ANN, however, is often prone to overfitting when the data comprise a limited number of observations. A Bayesian treatment to an ANN offers a practical solution to this problem by naturally allowing for regularization \citep{mackay1992,neal2012}. A Bayesian neural network (BNN) can also deal with model complexity, e.g., by selecting the number of hidden neurons in the model. In particular, a BNN treats the network weights to be random and obtains a posterior distribution over them \citep{barberetal1998, kendall2017uncertainties}. \par
Although CART, BCART, ANN, and BNN individually perform well, they exhibit certain drawbacks. Tree-based models may overfit the training data, or stick to local minima in the decision boundaries. Additionally, the training of neural networks suffers considerably in a limited-data setup. Thus, a hybrid (or ensemble) formulation of trees and neural networks can leverage their strengths and overcome their limitations. Several such hybrid models blending CART and ANNs have been discussed in the literature \citep{utgoff1989,sethi1990, sirat1990neural, kijsirikuletal2001, michelonietal2012, vanli2019nonlinear, chakraborty2019radial, chakraborty2019, chakraborty2019novel, chakraborty2018novel, chakraborty2020superensemble}, and have been useful for improving the prediction accuracy of the individual models. These hybrid models, however, only consider frequentist implementations of their components. Some other works have explored hybrid frequentist-Bayesian models in the context of parametric inference, hypothesis testing, and other inferential problems \citep{yuan2009bayesian, bayarrietal2004, bickel2015blending}. However, we are not aware of any hybrid algorithms blending frequentist and Bayesian methods for nonparametric regression. Motivated by this, we propose a hybrid approach, called the Bayesian neural tree (BNT) model, for feature-selection-cum-prediction purposes. 
BNT model utilizes the built-in feature selection mechanisms of tree-based models (CART and BCART), along with the accuracy and flexibility of neural net (ANN and BNN), particularly in limited-data-size settings. The proposal can overcome the deficiencies of the component models, 
have fewer tuning parameters, and are easily interpretable. On the theoretical side, we prove the models' statistical consistency, which gives a theoretical guarantee of their robustness. Finally, we explore the performance of the BNT models using several standard regression data sets. \par
The remainder of this article is organized as follows. Section \ref{sec2} discusses the proposed BNT model. Section \ref{sec3} explores the statistical properties of the BNT model. The empirical performance of the models using real-life data sets is addressed in Section \ref{sec4}. Section \ref{sec5} concludes the paper with a discussion on the future scope of this work.

\section{Formulation of the BNT models \label{sec2}}
We begin by establishing notation. We assume that models are trained on $n$ observations, and that there are $d$ predictor variables. For data point $i$, where $1 \leq i \leq n$, let $Y_i$ denote the response variable, $\overline Y_i$ denote its mean value, and $\hat Y_i$ denote the final prediction obtained from a model. Let $X_i = (X_{i1},\dots,X_{id})^{'}$ denote the input vector for the $i^{th}$ data point, where $1 \leq i \leq n$.  We denote the training data as $L_n = \{Y_i,X_i\}_{i=1}^{n}$. In what follows, we omit the subscript $i$ for simplicity of notation.

\subsection{Overview of constituent models}
\subsubsection{CART and BCART}
A CART model consecutively divides the predictor space into multiple regions. The partitioning begins at the root node, followed by splits at each internal node. A splitting rule (i.e., a chosen predictor and a split threshold) for a node is determined based on the minimization of the mean squared error (MSE) in regression settings. For each node, a stopping criterion called `minsplit' is defined in terms of the minimum number of observations required in the node for further splitting. A node with less than `minsplit' samples are labeled as a terminal node. At a terminal node, the predictor space is not split any further. Every data point falls into a region defined at one of the terminal nodes, and predictions are made using the parameter local to that region. A fully grown tree is often pruned back via cross-validation or cost-complexity pruning to avoid overfitting. \par
To illustrate the Bayesian version of CART, we assume that a tree $T$ has $b$ terminal nodes. Let the set of terminal node parameters be $\Lambda = \{\lambda_1,\dots,\lambda_b\}$. A prior is then placed on $(\Lambda,T)$ as 
\begin{align}
P(\Lambda,T) = P(\Lambda|T) \text{ } P(T),
\end{align}
where $P(T)$ is specified as a tree generating stochastic process comprising two functions, namely $P_{split}(m,T)$, the probability that a terminal node $m$ in a tree $T$ is split, and $P_{rule}(\gamma|m,T)$, the probability that a splitting rule $\gamma$ is assigned if $m$ is split \citep{chipman1998}. A general form of $P_{split}(m,T)$ is \citep{chipman1998}
\begin{align}
P_{split}(m,T) = \alpha (1+D_m)^{-\beta}, \label{eqsplit}
\end{align}
where $D_m$ denotes the number of splits before the $m^{th}$ node, and $0<\alpha<1$ and $\beta\geq0$. Larger values of $\beta$ make the splitting of deeper nodes less probable, since the RHS in \eqref{eqsplit} is a decreasing function of the depth $D_m$ of a node. The prior $P_{rule}(\gamma|m,T)$ is specified so that at an internal node, each available predictor is equally likely to be chosen for a split, and for a chosen predictor, each of its observed values is equally likely to be chosen as a splitting threshold.
$P(\Lambda|T)$ is generally specified so that the marginalization 
\begin{align}
P(Y|T,X) = \int P(Y|X, \Lambda, T) \text{ } P(\Lambda|T) \text{ } d\Lambda
\end{align}
is feasible \citep{chipman1998}. For a continuous $Y$, we model the values in the $m^{th}$ terminal node as a Gaussian with mean $\mu_m$ and variance $\sigma_m^2$, where $1 \leq m \leq b$ . Thus, we have $\Lambda = \big\{\mu_m, \sigma^2_m\big\}_{m=1}^{b}$, with $\mu_m$ and $\sigma^2_m$ having conjugate Gaussian and Inverse-Gamma priors respectively, as in \citet{chipman1998,chipman2002bayesian}.
The posterior over the possible tree models $P(T|Y,X)$ is analytically explored via a Metropolis-Hastings search algorithm. A `good' tree is usually found as a tradeoff between the number of terminal nodes $b$, and a high value of the marginal probability $P(Y|T,X)$.

\subsubsection{ANN and BNN} \label{bnnexplain}
An ANN is a nonparametric model consisting of an input layer, a certain number of hidden layers, and an output layer. All inputs to the network pass through the hidden layers, after which they are mapped to the final output. Each interconnection of neurons in an ANN is associated with a weight. In frequentist settings, such weights are obtained by minimizing an error function and its gradient. \par
We consider an ANN with parameter vector $\theta$ and $\sigma$ denotes the variance function, which contains the network weights and a general offset (or bias) parameter. In the Bayesian setting, a zero-mean multivariate Gaussian prior is placed on $\theta$ \citep{mackay1992v2} as 
\begin{align}
    P(\theta) = \mathlarger{\frac{1}{\big(\frac{2\pi}{\sigma_p}\big)^{\frac{\mathlarger l}{2}}} \text{  exp} \big(-\frac{\sigma_p}{2} ||\theta||^2\big)},
\end{align}
where $l$ is the length of $\theta$. The likelihood is modeled as a Gaussian given by
\begin{align}
    P(L_n|\theta) = \mathlarger{\frac{1}{\big(\frac{2\pi}{\sigma_l}\big)^{\frac{\mathlarger n}{2}}} \text{  exp} \big(-\frac{\sigma_l}{2} \displaystyle \sum_{i=1}^{n} \big( \hat Y_i-Y_i^2\big)\big)}.
\end{align}
Predictions are obtained from the posterior predictive distribution
\begin{align}
P(Y|X, L_n) = \int_{\theta} P(Y|X, \theta) \text{ } P(\theta|L_n) \text{ } d\theta. \label{eqpost}
\end{align}
The integral in \eqref{eqpost} is approximated by $P(Y|X, \tilde\theta)$, where $ \tilde \theta$ is obtained by locally minimizing 
\begin{align}
    E = \mathlarger{\frac{\sigma_l}{2} \displaystyle \sum_{i=1}^{n} \big( \hat Y_i-Y_i^2\big) + \frac{\sigma_p}{2} ||\theta||^2}. \label{eqparts}
\end{align}
The first term in the RHS of \eqref{eqparts} corresponds to the error function that is minimized in frequentist settings. The second term corresponds to a regularization term that penalizes larger values in $ \theta $ and, hence, restrains overfitting. A BNN can also have a variable architecture, i.e., the number of hidden nodes can be subject to a Geometric distribution, enabling one to place a lower probability on more extensive networks, see \citet{insua1998feedforward}).

\subsection{Proposed BNT model}
We now describe the working principles of the proposed BNT model. We present two variants of BNT models where each consists of a Bayesian (frequentist) implementation of a tree-based component for feature selection purposes, and a frequentist (Bayesian) implementation of a neural network component for prediction purposes (see Figure \ref{figmodRWNGT}). Such hybridization of blending trees and neural networks in entirely frequentist settings were first proposed and theoretically justified in \citet{chakraborty2019radial,chakraborty2019,chakraborty2019novel}. In this work, we extend those approaches but consider frequentist and Bayesian versions of the component models. In theory, both BNT models are asymptotically consistent, as we prove in Section \ref{sec3}.
\subsubsection {BNT-1 model}
The BNT-1 model comprises two stages. In the first stage, a classical CART model is fit to the data, taking all $d$ predictors. The CART model implicitly selects a feature at each internal split (based on maximum reduction in the MSE). Thus, the features used to construct the CART model can be considered as `important' features in the data. We record these features, as well as the predictions obtained from the CART model. In the second stage, we construct a BNN with one hidden layer, where the input variables are the selected features from CART plus the prediction results from phase one. We use a Gaussian prior for the network weights and also model the data likelihood to be Gaussian. The prior for the number of hidden neurons ($k$) is taken to be a Geometric distribution with probability of success $p$. As illustrated in Section \ref{bnnexplain}, the BNN is naturally regularized through its implementation, hence making overfitting less likely. The final set of predictions is obtained after fitting the BNN model to the data. \par
Thus, the proposed BNT-1 model utilizes the intrinsic feature selection ability of CART in the first stage. It also trains a BNN model in the second stage using the selected features and predicted values from CART. This improves the accuracy of the individual models, as utilizing the CART output as a feature in the BNN adds non-redundant information. We present a formal workflow of the BNT-1 model below.
\begin{algorithm}[h!]
\KwIn{$L_n = \{Y;X_1,\dots,X_d\}$}
\KwOut{$\hat Y$} 
\nl  Fit a CART model to $L_n$ with a specified `minsplit' value.  
 \begin{itemize}
  \item Record $S \subseteq \{X_1,\dots,X_d\}$, the set of selected features from CART.
  \item Record $\hat Y_{cart}$, the predictions from CART.
  \item Construct $S^{'} = \{S,\hat Y_{cart}\}$, the complete set of features for the BNN model.
  \end{itemize}
\nl Fit a BNN model with $k$ hidden neurons, where $k\sim \text{Geometric }(p)$, and with input feature set $S^{'}$.
 \begin{itemize}
  \item Record $\hat Y$, the final set of predictions from the BNN.
  \end{itemize}
  \caption{{\bf BNT-1}}
\end{algorithm}

\begin{figure}[t]
\centering
\includegraphics[scale=0.75]{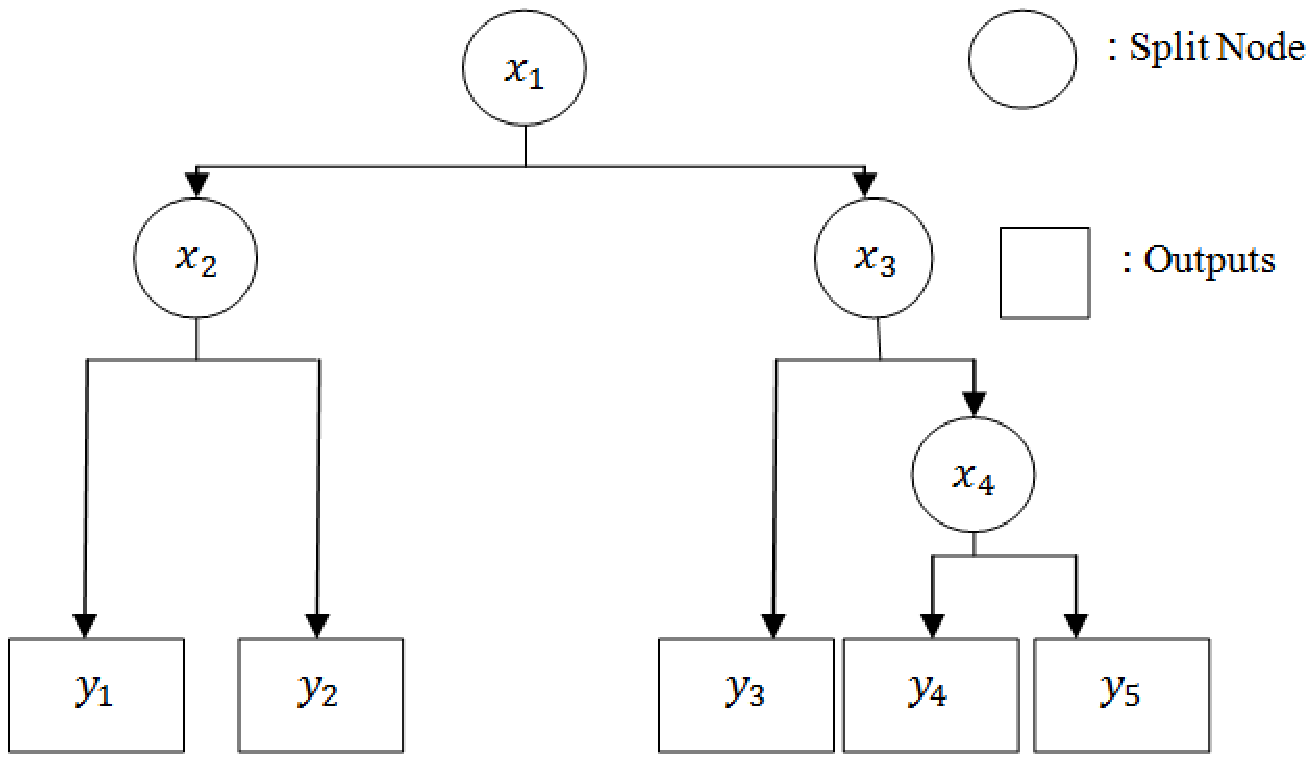}
\includegraphics[scale=0.75]{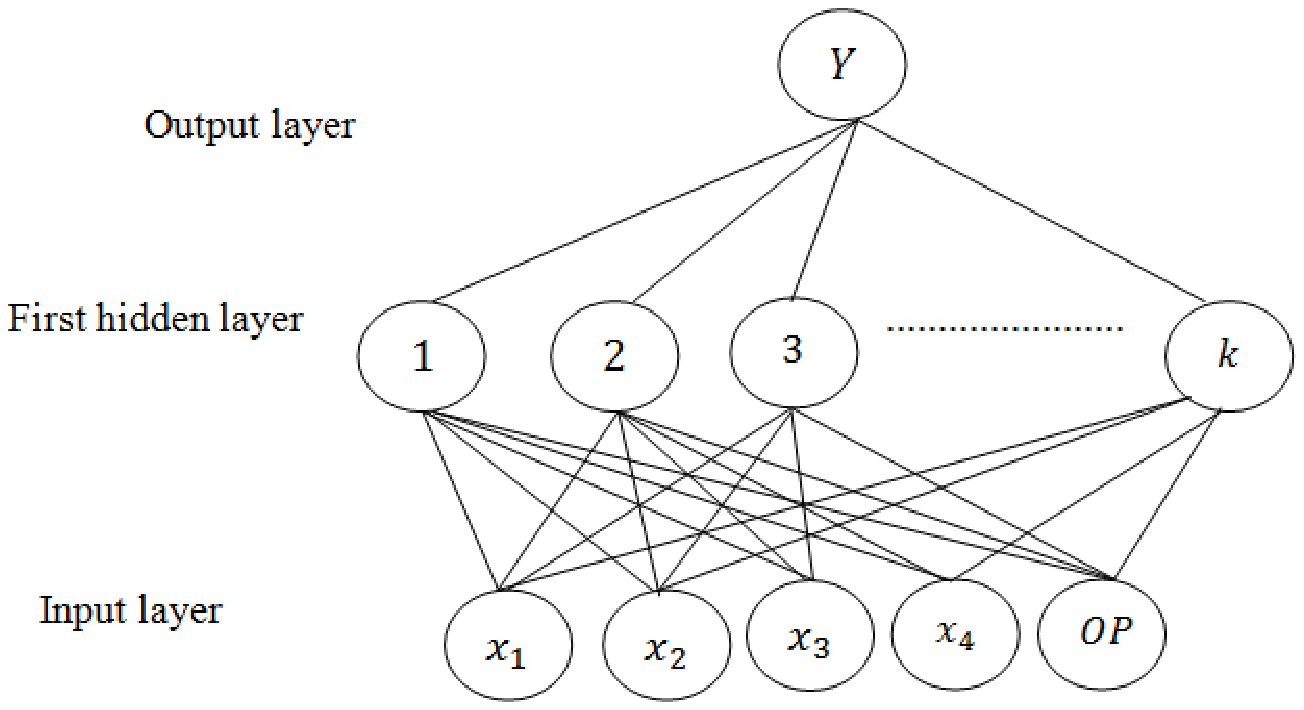}
\caption{An overview of Bayesian neural tree models. A CART (BCART) model is at the top and
its corresponding BNN (ANN) model at the bottom. ${\it OP}$ denotes the tree (CART/BCART) output.}
\label{figmodRWNGT}
\end{figure}

\subsubsection{BNT-2 model}
The BNT-2 model also follows a two-step pipeline. A BCART model fits the data in the first stage, with the best fitting tree found via posterior stochastic search. For feature selection in the context of BCART, \citet{bleich2014variable} illustrate three different schemes based on {\it variable inclusion proportions}, or the proportion of times a predictor variable is used for a split within each posterior sample. The three schemes differ in thresholding the inclusion proportions: `local', `global max', and `global SE' procedures. Any of the procedures can be utilized for feature selection based on the data and prediction problem at hand. In this work, we use the local thresholding procedure.\par
Thus, we record the important features and predictions from BCART and use these as inputs to a one-hidden-layer ANN in stage two. One hidden layer in the ANN sufficed, due to the incorporation of the selected features and predicted outputs from BCART. Using a single hidden layer also reduces the overall complexity of the model and the risk of overfitting in small and medium-sized data sets \citep{devroye2013probabilistic}. The optimal choice for the number of hidden neurons ($k$) for the ANN is derived under Proposition \ref{optimalnumberann} in Section \ref{secoptimalval}, and is given as $\sqrt{\frac{n}{d_{m}log(n)}}$, where $d_m$ is the dimension of the input feature space of the ANN, and $n$ is the training sample size. The final set of predictions is obtained after fitting the ANN model to the data. The precise algorithm is as follows.
\begin{algorithm}[h!]
\KwIn{$L_n = \{Y;X_1,\dots,X_d\}$}
\KwOut{$\hat Y$} 
\nl  Fit a BCART model to the data via a posterior stochastic search over the possible tree models.  
 \begin{itemize}
  \item Record $S \subseteq \{X_1,\dots,X_d\}$, the set of selected features obtained using a thresholding procedure.
  \item Record $\hat Y_{bcart}$, the prediction from BCART.
  \item Construct $S^{'} = \{S,\hat Y_{bcart}\}$, the complete set of features for the ANN model. Denote the dimension of $S^{'}$ as $d_m$.
  \end{itemize}
\nl Fit a one-hidden-layer ANN model with input feature set $S^{'}$, and with number of hidden neurons $k = \sqrt{\frac{n}{d_{m}log(n)}}$.
 \begin{itemize}
  \item Record $\hat Y$, the final set of predictions from the ANN.
  \end{itemize}
    \caption{{\bf BNT-2}}
\end{algorithm}

\section{Statistical Properties of the BNT models \label{sec3}}
From the results on the consistency of multivariate histogram-based regression
estimates on data-dependent partitions
\citep{nobel1996histogram, lugosi1996consistency}, and that of
regression estimates realized by an ANN \citep{lugosi1995nonparametric, devroye2013probabilistic}, we know that under certain conditions, both nonparametric models converge to the true density functions.
In Bayesian settings, posterior concentration of the BCART model \citep{rockova2017posterior}, and posterior consistency of the BNN model \citep{lee2000consistency, lee2004bayesian} have been previously explored. We use these results to prove the theoretical consistency of the BNT models
under certain conditions. We also find the optimal value of the number of hidden nodes in the BNT-2 model in Subsection \ref{secoptimalval}.

\subsection{Consistency of the BNT-1 Model}
Let $\mathbb{X}=(X_{1},X_{2},...,X_{d})$ be the space of all
possible values of $d$ features, and let $\mathbb{Y} = (Y_1,\dots,Y_n)^{'}$ be the response
vector, where each $Y_i$ takes values in $[-K,K]$, and $K \in \mathbb{R}$. A
regression tree (RT) $f : \mathbb{R}^{d} \rightarrow \mathbb{R}$ is
defined by assigning a number to each cell of a tree-structured
partition. We seek to estimate a regression function $r(x)= E(Y | X =
x) \in [-K,K]$ based on $n$ training samples $L_{n}=\{(X_{1},Y_{1}),
(X_{2},Y_{2}),...,(X_{n},Y_{n})\}$. The regression function $r(x)$
minimizes the predictive risk $J(f)=E\big|f(X) - Y\big|^2$ over all
functions $f: \mathbb{R}^d \rightarrow \mathbb{R}$. Practically, given the training data $L_n$, we can likely
find an estimate $\hat{f}$ of $f$ that minimizes the empirical risk
$$ J_n(f) = \frac{1}{n}\sum_{i=1}^{n}\Big(f(X_i)-Y_{i}\Big)^{2}$$
over a suitable class of regression estimates, since the distribution
of $(\mathbb{X},\mathbb{Y})$ is not known a priori. We let
$\Omega=\{\omega_{1},\omega_{2},...,\omega_{k}\}$ be a partition of
the feature space $\mathbb{X}$ and denote $\widetilde{\Omega}$ as
one such partition of $\Omega$. Define
$(L_{n})_{\omega_{i}}=\{(X_{i},Y_{i})\in L_{n}: X_{i}\in \omega_{i},
Y_{i}\in [-K,K]\}$ to be the subset of $L_{n}$ induced by
$\omega_{i}$ and let $(L_{n})_{\widetilde{\Omega}}$ denote the
partition of $L_{n}$ induced by $\widetilde{\Omega}$. 
Now define $\widehat{L}_{n}$ to be the space of all learning
samples and $\mathbb{D}$ be the space of all partitioning regression
functions. Then a binary partitioning rule
$f:\widehat{L}_{n} \rightarrow \mathbb{D}$ is such that $f \in (\psi
\circ \phi)(L_{n})$, where $\phi$ maps $L_{n}$ to some induced
partition $(L_{n})_{\widetilde{\Omega}}$ and $\psi$ is an assigning
rule which maps $(L_{n})_{\widetilde{\Omega}}$ to a partitioning
regression function $\hat{f}$ on the partition $\widetilde{\Omega}$.
Consistent estimates of $r(\cdot)$ can be achieved using an
empirically optimal regression tree if the size of the tree grows
with $n$ at a controlled rate.

\begin{theorem}
Suppose $(\mathbb{X},\mathbb{Y})$ is a random vector in
$\mathbb{R}^{p} \times [-K,K]$ and $L_{n}$ is the training set of $n$
outcomes. Finally, for every $n$
and $w_{i}\in \tilde{\Omega}_{n}$, the induced subset
$(L_{n})_{w_{i}}$ contains at least $k_{n}$ of the vectors of
$X_{1},X_{2},...,X_{n}$. Let $\hat{f}$ minimizes the empirical risk
$J_n(f)$ over all $k_n$ nodes of RT $f \in (\psi \circ
\phi)(k_{n})$. If $k_n \rightarrow \infty$ and $k_n = o
\Big(\frac{n}{log(n)}\Big)$, then $P \big|\hat{f} - r \big|^2 \rightarrow 0$
with probability 1. \label{theorem100}
\end{theorem}

\begin{proof}
For proof, one may refer to \citet[Theorem 1]{chakraborty2019novel}.
\end{proof}

The BNT-1 model essentially uses the feature selection
mechanism of RT and RT output also plays an important role in
designing the ensemble model. We further build a one hidden layered
BNN model using RT given features as well as RT output as an another
feature in the input space of BNN. We denote the dimension of the
input feature space of BNN model in the ensemble as $d_m \; (\leq
d)$. We further assume that these covariates are fixed and have been rescaled to $[0,1]^{d_m} = \mathbb{C}^{d_m}$.Now, let the random variables $\mathbb{Z}$ and $\mathbb{Y}$
take their values from $\mathbb{C}^{d_{m}}$ and $[-K,K]$
respectively. Denote the measure of $\mathbb{Z}$ over
$\mathbb{C}^{d_{m}}$ by $\mu$ and $m:\mathbb{C}^{d_{m}}\rightarrow
[-K,K]$ be a measurable function such that $m(Z)$ approximates $Y$.
Given the training sequence $\xi_{n}=\{ (Z_{1},Y_{1}),
(Z_{2},Y_{2}),...,(Z_{n},Y_{n}) \}$ of $n$ iid copies of
($\mathbb{Z},\mathbb{Y}$), the parameters of the neural network
regression function estimators are chosen such that it minimizes the
empirical $L_{2}$ risk = $\frac{1}{n}
\sum_{j=1}^{n}|f(Z_{j})-Y_{j}|^{2}$. We have used logistic squasher
as sigmoid function in BNN and treat the number of hidden nodes
($k$) as a parameter in the proposed Bayesian ensemble formulation.
In usual Bayesian nonparametrics, number of hidden nodes grows with
the sample size and thus we can use an arbitrarily large number of
hidden nodes asymptotically. But we use the formulation by
\citet{insua1998feedforward} and treat number of hidden nodes in the
ensemble model as a parameter and show that the joint posterior
becomes consistent under certain regularity conditions. Following
\citet{insua1998feedforward} we consider geometric prior for $k$. This will give better uncertainty quantification by allowing unconstrained size of
the hidden nodes. The major advantage of using Bayesian setting over
frequentist approach is that it allows one to use background
knowledge to select a prior probability distribution for the model
parameters. Also the predictions of the future observations are made
by integrating the model's prediction with respect to the posterior
parameter distributions obtained by updating the prior by taking
into account the data. We address this by properly defining the
class of prior distribution for neural network parameters that reach
sensible limits when the size of the networks goes to infinity and
further implementing markov chain monte carlo algorithm in the
network structure \citep{mackay1992v2}. We define
\begin{equation}
E\big[Y_i | Z_i = z_i\big] = \beta_0 +
\sum_{j=1}^{k}\beta_j\sigma(a_j^{T}z_i)+\epsilon_i, \label{1}
\end{equation}
where $k$ is the number of hidden nodes, $\beta_j$'s are the weights
of these hidden nodes, $a_j$'s are vectors of location and scale
parameters, and $\epsilon_i \stackrel{iid}{\sim}
\mathcal{N}(0,\sigma^{2})$. Expanding (\ref{1}) in vector notation yields the following equation:
\begin{equation}
y_i = \beta_0 +
\sum_{j=1}^{k}\beta_j\sigma\bigg(a_{j0}+\sum_{h=1}^{d_m}a_{jh}z_{h}\bigg)+\epsilon_i,
\label{2}
\end{equation}
where $d_m$ is the number of input features. We consider the asymptotic properties of neural network in the Bayesian setting. We show
consistency of the posterior for neural networks in Bayesian
setting which along with Theorem \ref{theorem100} ensures the consistency of the proposed BNT-1 model. 

Let $\lambda_i = P(k=i)$ be the prior probability that the number of
hidden nodes is $i$, and of course $\sum_{i} \lambda_i = 1$. Also, $\Pi_i$
be the prior for the parameters of the regression equation, given
that $k = i$. We can then write the joint prior for all the
parameters as $\sum \lambda_i\Pi_i$. Here we consider $\Pi_i
\stackrel{ind}{\sim} \mathcal{N}(0,\tau^{2})$ and the prior for $k$
be geometric distribution. In the sequel, we also assume that $$ Y | Z = z \sim \mathcal{N}\Bigg( \beta_0 + \sum_{j=1}^{k}\frac{\beta_j}{1+\exp\Big({-a_{j0}-\sum_{h=1}^{d_m}a_{jh}z_{h}}\Big)}, 1 \Bigg). $$
Let $f_0(z,y)$ be the true density. We can define a family of Hellinger neighborhoods as
$$
H_\epsilon = \{f\; ; \; D_H(f,f_0)\leq \epsilon\},
$$
with $D_H(f,f_0)$ as defined below:
$$
D_H(f,f_0) = \mathlarger{\sqrt{\int \int \Bigg(\sqrt{f(z,y)} - \sqrt{f_0(z,y)}\Bigg)^2dz dy}}.
$$
Let $\mathscr{F}_{n}$ be the set of all neural networks with parameters $|a_{jh}|\leq C_n$ and  $|\beta_j|\leq C_n$, where $j=1,\dots,k$ and $h=1,\dots,d_m$, and $C_n$ grows with $n$ such that $C_n \leq \text{exp } (n^{(b-a)})$ for any constant $b$ such that $0<a<b<1$ when $k\leq n^a$. The Kullback-Leibler divergence (not a distance metric) is defined as 
$$
D_K(f_0,f) = E_{f_0}\Bigg[\text{log } \frac{f_0(z,y)}{f(z,y)}\Bigg].
$$
For any $\gamma > 0$, we define Kullback-Leibler neighborhood by $$ K_\gamma = \{ f \; : D_K(f_0,f) \leq \gamma \}. $$ 
We denote the prior for $f$ by $\Pi_n(\cdot)$ and the posterior by $P\big(\cdot \; |
(Z_1,Y_1), ...,(Z_n,Y_n)\big).$ Now we are going to
present results on the asymptotic properties of the posterior
distribution for the neural network model present in the ensemble BNT-1 model over Hellinger neighborhoods.

\begin{theorem}
Assume that $\mathbb{Z}$ is uniformly distributed in $[0,1]^{d_m}$, $\Pi_i
\stackrel{ind}{\sim} \mathcal{N}(0,\tau^{2})$, $k \sim
\mbox{Geometric}$, and the following conditions hold:\\
(A1) For all $i$, we have $\lambda_i > 0$; \\
(A2) $B_n \uparrow n$, for all $r>0$, there exists $q>1$ and $N$
such that
$\sum_{i=B_n+1}^{\infty} \lambda_i < exp\big( -n^q r \big)$ for $n \geq N$; \\
(A3) There exists $r_i > 0, N_i$ such that
$\Pi_n(\mathscr{F}_{n}^{c})
< exp(-nr_i)$ for all $n \geq N_i$;\\
(A4) For all $\gamma, v >0$, there exists $I$ and $M_i$ such that
for any $i \geq I$, $\Pi_i(K_\gamma)
\geq exp(-nv)$ for all $n \geq M_i$.\\
Then for all $\epsilon > 0$, posterior is asymptotically consistent
for $f_0$ over Hellinger neighborhoods, i.e.,\\ $P\big(H_\epsilon \; |
(Z_1,Y_1), ...,(Z_n,Y_n)\big) \rightarrow 1$ in probability.
\label{theo2}
\end{theorem}

\begin{proof}
To prove the theorem, we first show that the regularity conditions hold when we assume a Geometric prior for $k$. And finally, show the posterior consistency by using conditions (A1)-(A4).\par
Since we take geometric prior for $k$, it is obvious that
$\lambda_i > 0$.
\begin{align}
& \mbox{Now,} \; \sum_{i=B_n+1}^{\infty} \lambda_i = P(k > B_n) =
\sum_{i=B_n+1}^{\infty} p(1-p)^i = (1-p)^{B_n+1} \nonumber \\
& \quad \quad \quad \quad = exp\big[(B_n+1)log(1-p)\big] \nonumber \\
& \quad \quad \quad \quad = exp\bigg[-n^q\big(-log(1-p)\big)\bigg] \quad \mbox{(Using $B_n = O(n^q)$ for $q>1$)} \nonumber \\
& \quad \quad \quad \quad \leq exp\big(-n^q.r\big) \quad \mbox{for
$r > 0$ and sufficiently large $n$} \label{2star}
\end{align}
We consider a geometric prior with parameter $p$. Also let,
$B_n = O(n^q)$ for any $q>1$. For any $i$, we write $i < n^a$ for $a
> 0$ and sufficiently large $n$, where $\theta$ be the vector of all
parameters (other than $k$):
\begin{align}
& \Pi_{i}\big(\mathscr{F}_n^c\big) = \int_{\mathscr{F}_n^c}\Pi_i(\theta)d\theta \nonumber \\
& \quad \quad \quad \quad \leq \sum_{i=1}^{d_n} 2
\int_{C_n}^{\infty} \phi \bigg( \frac{\theta_i}{\tau_i} \bigg)
d\theta_{i} \nonumber \\
& \quad \quad \quad \quad = d_n \bigg[ 2\tau
\int_{\frac{C_n}{\tau_i}}^{\infty} \phi(\tau_i)d\tau_i  \bigg] \nonumber \\
& \quad \quad \quad \quad \leq d_n \bigg[ \frac{2 \tau \phi
\big(\frac{C_n}{\tau_i}\big)}{\frac{C_n}{\tau_i}} \bigg] \quad
\mbox{by Mill's ratio} \nonumber \\
& \quad \quad \quad \quad = d_n \bigg[
\frac{2\tau_i^2}{C_n}.\frac{1}{\sqrt{2\Pi}}.exp\bigg( -
\frac{C_n^2}{2\tau_i^2} \bigg) \bigg] \nonumber \\
& \quad \quad \quad \quad \leq d_n
\bigg(\tau_i^2\sqrt{\frac{2}{\Pi}}\bigg). exp\bigg( -n^{b-a} -
\frac{1}{2\tau_i^2}e^{2n^{b-a}} \bigg) \quad \big[ \mbox{Taking}
\quad C_n=e^{n^{b-a}},
0<a<b<1 \big] \nonumber \\
& \quad \quad \quad \quad = exp\bigg[ -n^{b-a} + \log\bigg(
d_n\tau_i^2 \sqrt{\frac{2}{\Pi}} \bigg) \bigg].exp \bigg(
- \frac{1}{2\tau_i^2}e^{2n^{b-a}} \bigg) \nonumber \\
& \quad \quad \quad \quad \leq exp \bigg( -
\frac{1}{2\tau_i^2}e^{2n^{b-a}} \bigg) \quad \big[ \mbox{Using} \quad d_n = (p+2)n^a+1 \leq (p+3)n^a \quad \mbox{for large $n$} \big] \nonumber \\
& \quad \quad \quad \quad \leq e^{-nr_i}, \quad \mbox{where} \quad
e^{2n^{b-a}} > n \quad \mbox{for large $n$ and taking} \quad
r=\frac{1}{2\tau_i^2}. \label{star}
\end{align}
We can write $$\mathscr{G}_n^c = \bigcup_{i=0}^{\infty}
\mathscr{F}_i^c,$$ where $\mathscr{F}_i$ be the set of all neural
networks with $i$ nodes and with all the parameters less than
$C_n$ in absolute value, $C_n \leq exp(n^b), 0<b<1.$
$$ \Pi(\mathscr{G}_n^c) = \sum_{i=0}^{\infty}
\lambda_i \Pi_i(\mathscr{F}_i^c) \leq \sum_{i=0}^{B_n}\lambda_i
\Pi_i(\mathscr{F}_i^c) + \sum_{i=B_n+1}^{\infty}\lambda_i = I_1 +
I_2.$$
To handle $I_1$ and $I_2$, we use (\ref{star}) and (\ref{2star}):
\begin{align*}
& I_1 \leq \sum_{i=0}^{B_n}\lambda_i exp(-nr_i) \nonumber \\
& \quad \leq exp(-nr^{*})\sum_{i=0}^{B_n}\lambda_i \quad \big( \mbox{By letting} \quad r^{*} = min\{r_0,r_1,...,r_{B_n}\} \big) \nonumber \\
& \quad \leq exp(-nr^{*}).
\end{align*}
And, $I_2 \leq exp(-n^qr^{*}) \quad \mbox{for sufficiently large
n.}$ For large $n, q>1, \quad \mbox{and} \quad r=r^{*}/2$, we have
$$\Pi(\mathscr{G}_n^c) < exp(-nr).$$
\begin{align}
& \Pi_{i}(M_\delta) = \prod_{i=1}^{\hat{d_n}}
\int_{\theta_i-\delta}^{\theta_i+\delta} \frac{1}{\sqrt{2\Pi
\tau^2}}.exp\bigg( -
\frac{u^2}{2\tau^2} \bigg)du \nonumber \\
& \quad \quad \quad \geq \prod_{i=1}^{\hat{d_n}} 2\delta \inf_{u \in
[\theta_i-1, \theta_i+1]} \frac{1}{\sqrt{2\Pi \tau^2}}.exp\bigg( -
\frac{u^2}{2\tau^2} \bigg) \nonumber \\
& \quad \quad \quad = \prod_{i=1}^{\hat{d_n}} \delta
\sqrt{\frac{2}{\Pi \tau^2}}.exp\bigg( - \frac{\xi_i}{2\tau^2} \bigg)
\quad \big[\xi_i = max\{ (\theta_i-1)^2, (\theta_i+1)^2 \} \big]
\nonumber \\
& \quad \quad \quad \geq \bigg( \delta \sqrt{\frac{2}{\Pi \tau^2}}
\bigg)^{\hat{d_n}}.exp\bigg( - \frac{\hat{\xi}\hat{d_n}}{2\tau^2}
\bigg) \quad \big[\hat{\xi} = \max_i\{ \xi_1, \xi_2, ... ,
\xi_{\hat{d_n}} \} \big]
\nonumber \\
& \quad \quad \quad > e^{-nv} \quad \big[ \mbox{Using} \quad
\hat{d_n} \leq (p+3)n^a \quad \mbox{and for large $n$ and for any
$v$} \big]. \label{3star}
\end{align}
For any $\delta > 0$, let $l$ be the number of hidden nodes required
by the theorem for making $g_0$ continuous and square
differentiable. Using (\ref{3star}) we write
$$ \Pi(M_\delta) = \sum_{i=0}^{\infty}
\lambda_i \Pi_i(M_\delta) \geq \lambda_l \Pi_l(M_\delta) \geq
\lambda_l exp(-nv^{*}).$$ For sufficiently large $n$ and for any
$v^{*}$, $l$ is a constant, thus $\lambda_l$ does not depend on $n$
and is positive for geometric prior. Thus, $\Pi(M_\delta) \geq
\exp(-nv)$ for any sufficiently large $n$.

We can now use conditions (A1)-(A4) to show that $P\big(H_\epsilon \; |
(Z_1,Y_1), ...,(Z_n,Y_n)\big) \rightarrow 1$ in probability. Alternatively, $P\big(H_\epsilon^{c} \; |
(Z_1,Y_1), ...,(Z_n,Y_n)\big) \rightarrow 0$ in probability. Now,
\begin{align*}
& P\big(H_\epsilon^{c} \; | (Z_1,Y_1), ...,(Z_n,Y_n)\big) = \frac{\int_{H_\epsilon^{c}}\prod_{i=1}^{n}f(z_i,y_i)d\Pi_n(f)}{\int\prod_{i=1}^{n}f(z_i,y_i)d\Pi_n(f)} \\
& \quad \quad \quad \quad \quad \quad \quad \quad \quad \quad \quad \quad = \frac{\int_{H_\epsilon^{c}}R_n(f)d\Pi_n(f)}{\int R_n(f)d\Pi_n(f)}, \; \mbox{where} \; R_n(f) = \frac{\prod_{i=1}^{n}f(z_i,y_i)}{\prod_{i=1}^{n}f_0(z_i,y_i)} \\
& \quad \quad \quad \quad \quad \quad \quad \quad \quad \quad \quad \quad = \frac{D_1}{D_2}.
\end{align*}
Using \citet{wong1995probability} and (A1-A4), we can find the supremum of the likelihood ratios $R_n(f)$. Thus, we have $ D_1 < e^{\frac{-nt}{2}} + e^{-2c_2\epsilon^2}, \; t,c_2 > 0. $\\
Using \citet[Lemma 5]{lee2000consistency} along with (A1-A4) we have $ D_2 > e^{-n\delta} $ for large $n$, except on a set with probability approaching to $0$.\\
Finally, we have 
\begin{align*}
& P\big(H_\epsilon^{c} \; | (Z_1,Y_1), ...,(Z_n,Y_n)\big) < \frac{e^{\frac{-nt}{2}} + e^{-2c_2\epsilon^2}}{e^{-n\delta}} \\
& \quad \quad \quad \quad \quad \quad \quad \quad \quad \quad \quad \quad \leq e^{-n\epsilon^{'}} + e^{n\epsilon^2\epsilon^{'}}, \; \mbox{where} \; \epsilon^{'} > 0 \\
& \quad \quad \quad \quad \quad \quad \quad \quad \quad \quad \quad \quad \rightarrow 0 \; \mbox{for sufficiently large $n$}.
\end{align*}
\end{proof}
\begin{remark}
Theorem \ref{theo2} shows that the posterior is consistent when the number of hidden neurons of the neural network (with Bayesian setting) is a parameter that can be estimated from the data. Thus, we can let the data derive the number of hidden nodes in the model and emphasize on model selection during practical implementation.
\end{remark}

\subsection{Consistency and optimal value of a parameter for the BNT-2 model} \label{secoptimalval}
We consider the nonparametric regression model 
\[
Y_i=f_0(X_i) + \epsilon_i, \; \; \epsilon_i \stackrel{iid}{\sim} \mathcal{N}(0,1),
\]
where the output variable $\mathbb{Y}=(Y_1,Y_2,...,Y_n)^{'}$ is dependent on a set of $d$ potential covariates $\mathbb{X}= (X_{i1},X_{i2},...,X_{id})^{'}$, $1 \leq i \leq n$. We further assume that these covariates are fixed and have been rescaled such that every $X_{ij} \in [0,1]^{d}=\mathbb{C}^{d}$, $1 \leq i \leq n$ and $1 \leq j \leq d$. The true unknown response surface $f_0(X_i)$ is assumed to be smooth. In a recent work \citep{rockova2017posterior}, it was shown that the BCART model achieves a near-minimax-rate optimal performance when approximating a single smooth function. Thus, optimal behavior of a BCART model is guaranteed, and even under a suitably complex prior on the number of terminal nodes, a BCART model is reluctant to overfit. In the BNT-2 model, we build a BCART model in the first stage, and perform variable (feature) selection as in \citet{bleich2014variable}, which ensures that we can obtain a consistent BCART model under the assumptions of Theorem 4.1 of \citet{rockova2017posterior}. The selected important features along with the BCART outputs are trained using an ANN model with one hidden layer. We denote the dimension of the input feature space of this ANN model as $d_m \; (\leq d)$. The rescaled feature space is denoted by $\mathbb{C}^{d_m} = [0,1]^{d_m}$. Using one hidden layer in the ANN makes the BNT-2 model less complex and fastens its actual implementation. Moreover, there is no theoretical 
gain in considering more than one hidden layer in an ANN \citep{devroye2013probabilistic}. Below, we establish sufficient conditions for consistency of the BNT-2 model along with the optimal value of the number of hidden nodes $k$.\par
Let the rescaled set of features of the ANN be $\mathbb{Z}$. $\mathbb{Z}$ and $\mathbb{Y}$ take
values from $\mathbb{C}^{d_{m}}$ and $[-K,K]$, respectively.
We denote the measure of $\mathbb{Z}$ over $\mathbb{C}^{d_{m}}$ by
$\mu$ and $m:\mathbb{C}^{d_{m}}\rightarrow [-K,K]$ be a measurable
function that approximates $\mathbb{Y}$. Given the training sequence $(\mathbb{Z}, \mathbb{Y})$ of $n$ i.i.d copies, the neural network hyperparameters are
chosen by empirical risk minimization. We consider the class of neural networks having a logistic sigmoidal activation function in the hidden layer and $k$ hidden neurons, with bounded output weights
\[
\mathscr{F}_{n,k}=\Bigg\{
\sum_{i=1}^{k}c_{i}\sigma(a_{i}^{T}z+b_{i})+c_{0} : k \in
\mathbb{N}, a_{i} \in \mathbb{R}^{d_{m}}, b_{i},c_{i} \in
\mathbb{R}, \sum_{i=0}^{k}|c_{i}|\leq \beta_{n} \Bigg\},
\]
and obtain $m_{n} \in \mathscr{F}_{n,k}$ satisfying
\[
\frac{1}{n} \sum_{i=1}^{n}|m_{n}(Z_{i})-Y_{i}|^{2}  \leq \frac{1}{n}
\sum_{i=1}^{n}|f(Z_{i})-Y_{i}|^{2}, \; \mbox{if} \; f \in
\mathscr{F}_{n,k},
\]
where, $m_{n}$ is a function that minimizes the empirical $L_{2}$
risk in $\mathscr{F}_{n,k}$. The theorem below, due to \citet[Theorem
3]{lugosi1995nonparametric}, states the sufficient conditions for the consistency of the neural network. 

\begin{theorem}
Consider an ANN with a logistic sigmoidal activation function having one hidden layer with $k \; (>1)$ hidden nodes. If $k$ and $\beta_{n}$ are chosen to satisfy $$k \rightarrow \infty, \; \beta_{n}\rightarrow \infty, \; \frac{k\beta_{n}^{4}log(k\beta_{n}^{2})}{n} \rightarrow 0$$ as $n \rightarrow \infty$, then the model is said to be consistent for all distributions of $(\mathbb{Z}, \mathbb{Y})$ with $\mathbb{E}|\mathbb{Y}|^2 < \infty$.\label{theo3}
\end{theorem}
\begin{proof}
For the proof, one may refer to \citet[Chapter 16]{gyorfi2006distribution}.
\end{proof}
Now, we obtain an upper bound on $k$ using the rate of convergence of a neural network with bounded output weights. In what follows, we have assumed that $m$
is Lipschitz $(\delta, C)$-smooth according to the following
definition:
\begin{definition}
A function $m:C^{d_m} \rightarrow [-K,K]$ is called Lipschitz $(\delta,
C)$-smooth if it satisfies the following equation: $$|m(z^{'})-m(z)|
\leq C\|z^{'}-z\|^\delta$$ for all $\delta \in [0,1]$, $z^{'}, z \in
\mathbb{C}^{d_m}$, and $C \in \mathbb{R}_{+}$.
\end{definition}

\begin{proposition} \label{optimalnumberann}
Assume that $\mathbb{Z}$ is uniformly distributed in $\mathbb{C}^{d_m}$ and $\mathbb{Y}$ is bounded a.s. and $m$ is Lipschitz $(\delta, c)$-smooth. Under the assumptions of Theorem (\ref{theo3}) with fixed $d_{m}$, and $m, f \in \mathscr{F}_{n,k}$, also $f$ satisfying
$\int_{C^{d_m}}f^{2}(z)\mu(dz)<\infty$, we have $k = O \bigg
(\sqrt{\frac{n}{d_{m}log(n)}} \bigg )$.
\end{proposition}

\begin{proof}
To prove Proposition \ref{optimalnumberann}, we use results from statistical learning theory of neural networks \citep[Chapter 12]{gyorfi2006distribution}. We use the complexity regularization principle to choose the parameter $k$ in a data-dependent manner \citep{kohler2005adaptive, hamers2003bound, kohler2006nonparametric}. Consistency results presented in Theorem \ref{theo3} state that
\[
\mathbb{E}\int_{C^{d_m}}(m_{n}(Z)-m(Z))^{2}\mu(dz) \rightarrow 0 \quad \mbox{as}\quad n
\rightarrow \infty.
\]
We can
write, using \citet[][Lemma 10.1]{gyorfi2006distribution}, that
\begin{align}
\mathbb{E}\Bigg[\int_{C^{d_m}} \big| m_n(Z)-m(Z)\big|^2 \mu(\mbox{dz})\Bigg] \leq 2\mathbb{E}\Bigg[ \sup_{f \in
\mathscr
F_{n,k}}\bigg|\frac{1}{n}\sum_{i=1}^n\big|Y_i-f(Z_i)\big|^2 
- \mathbb{E}\big|Y-f(Z)\big|^2\bigg| \Bigg] \nonumber \\
+ \mathbb{E} \Bigg[\inf_{f \in \mathscr F_{n,k}}
\int_{C^{d_m}}\big|f(Z)-m(Z)\big|^2\mu(\mbox{dz}) \Bigg],\label{imp}
\end{align}
where $\mu$ denotes the distribution of $\mathbb{Z}$. For the consistency
of the neural network model, the \textit{estimation error} (first term in the RHS of
\ref{imp}) and the \textit{approximation error} (second term in the
RHS of \ref{imp}) should tend to 0. To find the bound for $k$, we apply 
non-asymptotic uniform deviation inequalities and covering numbers
corresponding to $\mathscr F_{n,k}$. Assuming $Y$ is bounded as in
Theorem \ref{theo3}, we write (\ref{imp}) as
\begin{align}
\mathbb{E} \int_{C^{d_m}} \big| m_n(Z)-m(Z)\big|^2 \mu(\mbox{dz}) \leq 2
\min_{k \geq 1}\bigg\{pen_n(k)+\inf_{f \in \mathscr
F_{n, k}}\int_{C^{d_m}}\big|f(z)-m(z)\big|^2\mu(\mbox{dz})\bigg\}
+O\Big(\frac{1}{n}\Big).\label{roc}
\end{align}

We have assumed that for each $f \in \mathscr F_n$, $Y$ is bounded. Let $w_1^n=(w_1, w_2, \hdots, w_n)$ be
a vector of $n$ fixed points in $ \mathbb{R}^{d_m}$ and let $\mathscr H$ be a set
of functions from $ \mathbb{R}^{d_m} \to [-K,K]$. For every $\varepsilon >0$, we let
$\mathscr N(\varepsilon, \mathscr H,w_1^n)$ be the $L_1$
$\varepsilon$-covering number of $\mathscr H$ with respect to $w_1, w_2, \hdots, w_n$. $\mathscr N(\varepsilon, \mathscr H,w_1^n)$  is
defined as the smallest integer $N$ such that there exist functions
$h_1, \hdots, h_N: \mathbb{R}^{d_m}\to [-K,K]$ with the property that for every $h\in
\mathscr H$, there is a $j \in \{1, \hdots, N\}$ such that
$$\frac{1}{n}\sum_{i=1}^n\big|h(w_i)-h_j(w_i)\big|<\varepsilon.$$
Note that if $W_1^n = (W_1, W_2,\hdots, W_n)$ is a sequence of
i.i.d. random variables, then $\mathscr N(\varepsilon, \mathscr
H,W_1^n)$ is also a random variable. Now, let $W=(Z,Y)$,
$W_1=(Z_1,Y_1), \hdots, W_n=(Z_n,Y_n)$, and $C^{d_m} = [0,1]^{d_m}$, we
write
\begin{align*}
\mathscr H_n=\bigg\{h(z,y)& :=|y-f(z)\big|^2: (z,y)\in C^{d_m}\times
[-K,K]\mbox{ and }f\in \mathscr F_n\bigg\}.
\end{align*}
The functions in $\mathscr H_n$ will satisfy the following: $0\leq
h(z,y)\leq 2\beta_n^2 + 2K^2 \leq 4\beta_n^2.$ Using Pollard's
inequality \citep{gyorfi2006distribution}, we have, for arbitrary
$\varepsilon
>0$,
\begin{align}
& P\bigg\{ \sup_{f \in \mathscr
F_{n,k}}\Big|\frac{1}{n}\sum_{i=1}^n\big|Y_i-f(Z_i)\big|^2
- E\big|Y-f(Z)\big|^2\Big|>\varepsilon\bigg\}\nonumber\\
& = P\bigg\{ \sup_{h \in \mathscr H_n}\Big|\frac{1}{n}\sum_{i=1}^nh(W_i)- E(h(W))\Big|>\varepsilon\bigg\}\nonumber\\
& \leq 8 \mathbb{E} \bigg[\mathscr N\big(\frac{\varepsilon}{8}, \mathscr
H_n,W_1^n\big)\bigg]\exp\bigg({-\frac{n\varepsilon^2}{128(4\beta_n^2)^2}}\bigg).\label{zero}
\end{align}
Next, we try to bound the covering number $ \mathscr
N\big(\frac{\varepsilon}{8}, \mathscr H_n, W_1^n\big)$. Let us
consider two functions $h_{i}(z,y)= |y-f_{i}(z)|^2$ of $\mathscr
H_n$ for some $f_{i} \in \mathscr F_n$ and $i = 1,2$. We get
\begin{align*}
&\frac{1}{n}\sum_{i=1}^n\big|h_{1}(W_i)-h_{2}(W_i)\big|\\
& =\frac{1}{n}\sum_{i=1}^n\Big|\big|Y_i-f_{1}(Z_i)\big|^2-\big|Y_i-f_{2}(Z_i)\big|^2\Big|\\
& =\frac{1}{n}\sum_{i=1}^n\big|f_{1}(Z_i)-f_{2}(Z_i)\big|\times \big|f_1(Z_i)-Y_i+f_2(Z_i)-Y_i\big|\\
& \leq \frac{4\beta_n}{n}\sum_{i=1}^n\big|f_{1}(Z_i)-f_{2}(Z_i)\big|.
\end{align*}
Thus, if $\{h_1,h_2,...,h_l\}$ is an $\frac{\mathlarger \varepsilon}{8}$ packing of $\mathscr H_n$ on $W_1^n$, then $\{f_1,f_2,...,f_l\}$ is an
$\frac{\mathlarger \varepsilon}{32\beta_n}$ packing of $\mathscr F_n$.
\begin{equation}
\mbox{Thus,} \; \; \label{one} \mathscr
N\Big(\frac{\mathlarger \varepsilon}{8}, \mathscr H_n, W_1^n\Big)\leq \mathscr
N\Big(\frac{\mathlarger \varepsilon}{32\beta_n}, \mathscr F_n, Z_1^n\Big).
\end{equation}
The covering number $\mathscr N(\frac{\varepsilon}{32\beta_n},
\mathscr F_n, Z_1^n)$ can be upper bounded independently of $Z_1^n$
by extending the arguments of Theorem 16.1 of
\citet[]{gyorfi2006distribution}. We now define the following classes of functions:
$$ G_1=\{\sigma(a^{\top} z + b): a\in \mathbb{R}^{d_m}, b\in \mathbb{R}\},$$
$$ G_2=\{c\sigma(a^{\top} z + b): a\in \mathbb{R}^{d_m}, b\in \mathbb{R}, c \in [-\beta_n, \beta_n]\}.$$
For any $ \varepsilon > 0 $,
\begin{align*}
\mathscr N(\varepsilon, G_1,Z_1^n) & \leq 3
\bigg(\frac{2e}{\varepsilon} \log\frac{3e}{\varepsilon} \bigg)^{d_{m}+2}
\\ & = 3\Big(\frac{3e}{\varepsilon}\Big)^{2d_{m}+4}.
\end{align*}
Also, we get
\begin{align*}
\mathscr N(\varepsilon, G_2,Z_1^n) & \leq
\frac{4\beta_n}{\varepsilon}\mathscr
N\Big(\frac{\varepsilon}{2\beta_n}, G_1,Z_1^n\Big) \\ & \leq \Big(\frac{12 e \beta_n}{\varepsilon}\Big)^{2d_{m}+5}.
\end{align*}
We obtain the bound on the covering number of $\mathscr F_n$,
\begin{align}
\mathscr N(\varepsilon, \mathscr F_{n},Z_1^n) & \leq
\frac{2\beta_n}{\varepsilon}\mathscr
N\Big(\frac{\varepsilon}{k+1}, G_2,Z_1^n\Big)^{k} \nonumber\\
& \leq \Big( \frac{12 e
\beta_n(k+1)}{\varepsilon}\Big)^{(2d_{m}+5)k+1}.\label{eight}
\end{align}
According to (\ref{eight}), and for any $Z_1^n \in \mathbb{R}^{d_m}$, we have
\begin{align}
\mathscr N\Big(\frac{1}{n}, \mathscr F_{n,k},Z_1^n\Big) \leq
\Big(12en\beta_n(k+1)\Big)^{(2d_m + 5)k+1}. \label{final}
\end{align}
Using the complexity regularization principle we have
$$\sup_{Z_1^n}\mathscr N\bigg(\frac{1}{n}, \mathscr
F_{k,n},Z_1^n\bigg) \leq \mathscr N\bigg(\frac{1}{n}, \mathscr
F_{k,n}\bigg)$$ to be the upper bound on the covering number of
$\mathscr F_{k,n}$, and define for $w_{k} \geq 0$, $$ pen_n(k) =
\frac{\mbox{constant} \times K^2 \times log\mathscr N\big(\frac{1}{n}, \mathscr
F_{k,n}\big)+w_{k}}{n}$$ as a penalty term penalizing the
complexity of $\mathscr F_{k,n}$ \citep{kohler2005adaptive}. Thus
(\ref{final}) implies that $pen_n(k)$ is of the following form
with $w_{k} = 1$ and $\beta_n = \mbox{constant} < \infty$,
$$ pen_n(k) =
\frac{\mbox{constant} \times K^2 \times (2d_{m}+6)klog\big(12en\beta_n)+1}{n}=O\bigg(\frac{kd_{m}log(n)}{n}\bigg).$$
The approximation error $\inf_{f \in \mathscr
F_{k,n}}\int_{C^{d_m}}\big|f(z)-m(z)\big|^2\mu(\mbox{dz})$ depends on
the smoothness of the regression function. According to Theorem 3.4
of \citet{mhaskar1993approximation}, for any feedforward neural network with one hidden layer satisfying the assumptions of Proposition \ref{optimalnumberann}, we have
$$ \big|f(z)-m(z)\big| \leq \bigg(\frac{1}{\sqrt{k}}\bigg)^{\frac{\delta}{d_m}}$$ for all $z \in
[0,1]^{d_m}$. Thus, we have,
$$ \inf_{f \in \mathscr
F_{n, k}}\int_{C^{d_m}}\big|f(z)-m(z)\big|^2\mu(\mbox{dz}) = O\Big( \frac{1}{k} \Big).$$
Using (\ref{roc}), we have
\begin{align}
\mathbb{E} \int_{C^{d_m}} \big| m_n(Z)-m(Z)\big|^2 \mu(\mbox{dz}) \leq O\Bigg( \frac{kd_mlog(n)}{n} \Bigg) + O\Bigg( \frac{1}{k} \Bigg)
\end{align}
for sufficiently large $n$.\\
Now we can balance the approximation error with the bound on the
covering number to obtain the optimal choice of $k$ from which the assertion follows. 
\end{proof}
\begin{remark}
For practical purposes, we choose the number of hidden neurons in the BNT-2 model to be $k=\sqrt{\frac{n}{d_{m}log(n)}}$. 
\end{remark}

\section{Experimental evaluation \label{sec4}}
We now present applications of the two BNT models to real-life data sets, and evaluate them against their component regression models, namely a simple CART model, a simple BCART model, a one-hidden-layer ANN, and a one-hidden-layer BNN.
\subsection{Data}
We use regression data sets available on the UCI machine learning repository (\url{https://archive.ics.uci.edu/ml/datasets.html}). These data sets have a limited number of observations and high-dimensional feature spaces. As a part of the data cleaning process, we systematically eliminate all nonnumerical features and observations with missing values.  Table \ref{tab:data} summarizes the characteristics of the data sets.

\subsection{Performance metrics} \label{secperfmet}
For evaluating the BNT models, we use two absolute performance measures, viz. the mean absolute error (MAE) and the root mean squared error (RMSE), one relative measure, viz. the mean absolute percentage error (MAPE), and two goodness of fit measures, i.e., the coefficient of determination ($R^2$) and adjusted $R^2$. The metrics are defined as follows: 
\begin{enumerate}
    \item MAE = $\mathlarger{\frac{1}{n}\sum_{i=1}^{n} |Y_i- \hat Y_i|}$,
    \item MAPE = $\mathlarger{\frac{1}{n}\sum_{i=1}^{n} \Bigg|\frac{Y_i-\hat Y_i}{Y_i}\Bigg|}$,
    \item RMSE = $\mathlarger{{\sqrt{\frac{1}{n}\sum_{i=1}^{n}(Y_i-\hat Y_i)^2}}}$,
    \item $R^2$ = $\mathlarger{{1-\frac{\sum_{i=1}^{n}(Y_i-\hat Y_i)^2}{\sum_{i=1}^{n}(Y_i- \overline{Y})^2 }}}$,
    \item Adjusted $R^2$ = $\mathlarger{{1-\frac{(1-R^2)(n-1)}{(n-d-1)}}}$,
\end{enumerate}
\begin{table}[t]
\centering
\caption{Summary of the datasets used to evaluate the BNT models.\label{tab:data}}
\begin{tabular}{|c|c|c|}
\hline
\textbf{Dataset}             & \textbf{Number of observations (n)} & \textbf{Number of features (d)} \\ \hline
\multirow{2}{*}{AutoMPG}     & \multirow{2}{*}{398}                & \multirow{2}{*}{7}              \\
                             &                                     &                                 \\ \hline
\multirow{2}{*}{Housing}     & \multirow{2}{*}{506}                & \multirow{2}{*}{13}             \\
                             &                                     &                                 \\ \hline
\multirow{2}{*}{Power}       & \multirow{2}{*}{9568}               & \multirow{2}{*}{4}              \\
                             &                                     &                                 \\ \hline
\multirow{2}{*}{Crime}       & \multirow{2}{*}{1994}                & \multirow{2}{*}{101}            \\
                             &                                     &                                 \\ \hline
\multirow{2}{*}{Concrete}    & \multirow{2}{*}{1030}               & \multirow{2}{*}{8}              \\
                             &                                     &                                 \\ \hline
\end{tabular}
\end{table}
We note that lower values of MAE, MAPE, and RMSE, and higher values of $R^2$ and adjusted $R^2$ indicate better model performance.

\subsection{Implementation and results}
We shuffle the observations in each data set and split into training and test sets in the ratio 70:30. We carry out ten random train-test splits and report average results across all ten iterations. All models are fit on the training data, and evaluated on the test data. Experiments are carried out using $\textsf{R}$ (version 3.6.1). We fit a CART model using the $\textsf{rpart}$ package, with the stopping parameter `minsplit' set to 10\% of the training sample size. To fit a simple BNN, we use the $\textsf{brnn}$ package with the number of hidden layers set to one and the number of hidden neurons set to the default value (i.e., 2). The $\textsf{brnn}$ package implements a BNN with a Gaussian prior and likelihood, as discussed in Section \ref{bnnexplain}. To fit a simple, one-hidden-layer ANN, we make use of the $\textsf{neuralnet}$ package and set the number of hidden neurons to the default value (2). A Bayesian CART model is fit using the $\textsf{bartMachine}$ package \cite{bleich2016bartmachine}, with the number of trees set to one. For feature selection under BCART, we use local thresholding of the variable inclusion proportions, although empirical explorations show that results are not very sensitive to other thresholding methods. As seen in Tables \ref{tabresults1} and \ref{tabresults2}, the component models of the BNTs exhibit consistent results, and neural networks perform better than the tree-based models for a majority of the data sets. \par
We now turn to the implementation of the two BNT models. To implement BNT-1, we first record the selected features and predictions from the CART model, forming the set of features for the subsequent BNN model. Again, a CART model is trained with the stopping parameter `minsplit' set to 10\% of the training sample size. A one-hidden-layer BNN is then fit with the number of hidden neurons $k$ drawn from Geometric distributions with success probabilities $p \in \{0.3,0.6,0.9\}$. To implement BNT-2, we record important features and predictions from the BCART model and use these as inputs to the ANN model with one hidden layer. The number of neurons in the ANN is taken to be $\sqrt{\frac{n}{d_{m}log(n)}}$, which is the optimal number derived in Section \ref{secoptimalval}. Additionally, all data sets are min-max scaled to be in the $[0,1]$ range before training the neural network models. From Tables \ref{tabresults1} and \ref{tabresults2}, we observe that across all data sets, the proposed BNT models greatly improve the performance of their component models. We note that the BNT-2 model outperforms all others on most data sets. Consequently, we can expect the BNT predictions to be at least better than the individual model predictions, since cases, where further optimization is likely to have led to overfitting, are directly filtered out.

\begin{table}[h!]
\centering
\caption{Performance metrics for the evaluated models across different data sets. \label{tabresults1}}
\begin{tabular}{|c|c|L|c|c|c|c|c|}
\hline 
{\bf Data Set}                 & {\bf Model}   &        & \multicolumn{5}{|c|}{\bf Performance Metrics} \\ \hline
                         &    &Number of features used &MAE        &MAPE     &RMSE      &$R^2$   &adjusted $R^2$ \\ \hline
\multirow{9}{*}{AutoMPG} 
                         & CART             &3               &2.640      &0.120    &3.419     &0.834   &0.830          \\
                         & BCART            &3               &2.796      &0.117    &3.693     &0.806   &0.803          \\
                         & ANN              &7               &2.241      &0.0967   &3.164     &0.858   &0.850          \\
                         & BNN              &7               &2.253      &0.097    &3.123     &0.861   &0.854           \\
                         & BNT-1 ($p$=0.3)  &4               &2.111      &0.091    &3.016     &0.871   &0.867          \\
                         & BNT-1 ($p$=0.6)  &4               &2.110      &0.092    &3.013     &0.871   &0.870           \\
                         & BNT-1 ($p$=0.9)  &4               &2.119      &0.092    &3.018     &0.873   &0.870           \\
                         & BNT-2            &4               &2.081      & 0.090   &3.0333    &0.869   &0.868        \\ \hline 
\multirow{9}{*}{Housing} 
                         & CART             &3              &3.161      &0.163    &5.068     &0.696   &0.690           \\
                         & BCART            &4               &3.683      &0.194    &5.057     &0.697   &0.689           \\
                         & ANN              &13              &2.736      &0.132    &4.782     &0.729   &0.706           \\
                         & BNN              &13               &2.742      &0.132    &4.793     &0.704   &0.702           \\
                         & BNT-1 ($p$=0.3)  &4               &2.643      &0.129    &4.731     &0.735   &0.730           \\
                         & BNT-1 ($p$=0.6)  &4               &2.641      &0.128    &4.730     &0.735   &0.730           \\
                         & BNT-1 ($p$=0.9)  &4               &2.641      &0.128    &4.730     &0.735   &0.730           \\
                         & BNT-2            &5              &2.751      &0.134    &4.597     &0.750   &0.748        \\ \hline 
\multirow{9}{*}{Power}  
                         & CART             &2              &4.157      &0.009    &5.389     &0.901   &0.901           \\
                         & BCART            &2              &5.502      &0.008    &4.561     &0.929   &0.929           \\
                         & ANN              &4               &3.558      &0.008    &4.501     &0.937   &0.937           \\
                         & BNN              &4               &3.563      &0.007    &4.510     &0.940   &0.940           \\
                         & BNT-1 ($p$=0.3)  &3               &3.444      &0.008    &4.460     &0.932   &0.932           \\
                         & BNT-1 ($p$=0.6)  &3               &3.443      &0.008    &4.463     &0.932   &0.932           \\
                         & BNT-1 ($p$=0.9)  &3               &3.442      &0.008    &4.461     &0.932   &0.932           \\
                         & BNT-2            &3               &3.408      &0.007    &4.410     &0.934   &0.934   \\ \hline
\end{tabular}

\end{table}

\begin{table}[t]
\centering
\caption{Performance metrics for the evaluated models across different data sets. \label{tabresults2}}
\begin{tabular}{|c|c|L|c|c|c|c|c|}
\hline 
{\bf Data Set}                 & {\bf Model}   &        & \multicolumn{5}{|c|}{\bf Performance Metrics} \\ \hline
                         &                &Number of features used &MAE    &MAPE   &RMSE    &$R^2$   &adjusted $R^2$ \\ \hline
\multirow{9}{*}{Crime}  
                         & CART              &12                         &0.166  &0.435  &0.230   &0.399   &0.335           \\
                         & BCART             &15                          &0.186  &0.580  &0.231   &0.394   &0.250           \\
                         & ANN               &101                          &0.164  &0.442  &0.222   &0.443   &0.468           \\
                         & BNN               &101                          &0.167  &0.567  &0.290   &0.580   &0.580           \\
                         & BNT-1 ($p$=0.3)   &13                          &0.158  &0.395  &0.218   &0.463   &0.406           \\
                         & BNT-1 ($p$=0.6)   &13                          &0.154  &0.395  &0.218   &0.463   &0.406                 \\
                         & BNT-1 ($p$=0.9)   &13                          &0.158  &0.395  &0.218   &0.463   &0.406                \\
                         & BNT-2             &16                          &0.143  &0.367  &0.193   &0.578   &0.574  \\ \hline 
\multirow{9}{*}{Concrete}
                         & CART              &5                         &7.462  &0.286  &9.414   &0.694   &0.689           \\
                         & BCART             &3                          &7.909  &0.304  &10.064  &0.651   &0.649           \\
                         & ANN               &8                          &6.987  &0.235  &9.194   &0.709   &0.701          \\
                         & BNN               &8                          &6.043  &0.268  &7.676   &0.746   &0.842           \\
                         & BNT-1 ($p$=0.3)   &6                          &5.493  &0.194  &6.961   &0.833   &0.830           \\
                         & BNT-1 ($p$=0.6)   &6                          &5.492  &0.194  &6.950   &0.840   &0.830                  \\
                         & BNT-1 ($p$=0.9)   &6                          &5.493  &0.194  &6.961   &0.833   &0.830                \\
                         & BNT-2             &4                          &5.473  &0.178  &6.636   &0.879   &0.878          \\ \hline
                        
\end{tabular}
\end{table}

\section{Concluding remarks\label{sec5}}
In this work, we present two hybrid models that combine frequentist and Bayesian implementations of decision trees and neural networks. The BNT models are novel, first-of-their-kind proposals for nonparametric regression purposes. 
We find that the models perform competitively on small to medium-sized datasets compared to other state-of-the-art nonparametric models. Moreover, the BNT models have a significant advantage over purely frequentist hybridizations. A Bayesian approach to constructing a CART or an ANN model can check to overfit. A BCART model allows placing priors that control the depth of the resultant trees, and BNNs with Gaussian priors are inherently regularized. This prevents the need to tune multiple parameters via cross-validation manually. Thus, the proposed BNT models overcome the deficiencies of their component models and the drawbacks of using fully frequentist or fully Bayesian models. We also show that the BNT models are consistent, which ensures their theoretical validity. An immediate extension of this work will be to construct BNT models for classification problems. Another area of future work will be to extend the proposed approaches to survival regression frameworks. 
\section*{Data and code}
For the sake of reproducibility of this work, code for implementing the BNT models
is made available at
\url{https://github.com/gaurikamat/Bayesian_Neural_Tree}. The data for the experiments is obtained from \url{https://archive.ics.uci.edu/ml/datasets.html}.


\bibliographystyle{plainnat} 
\bibliography{bibliography.bib}

\end{document}